\RequirePackage{fix-cm}
\documentclass[smallextended]{svjour3}       % onecolumn (second format)
\smartqed  % flush right qed marks, e.g. at end of proof
 \usepackage{mathptmx}      % use Times fonts if available on your TeX system
\usepackage{graphicx,amsmath}
\usepackage{amssymb}
\newcommand{\C}{\mathcal{C}}
\newcommand{\F}{\mathbb{F}}
\newcommand{\Z}{\mathbb{Z}}
\newcommand{\sh}{\mathcal{S}}
\newcommand{\LL}{\mathcal{L}}
\newcommand{\MM}{\mathcal{M}}
\newcommand{\A}{\mathcal{A}}
\newcommand{\lw}[1]{\mathrm{w}_{\mathcal{L}}(#1)}
\newcommand{\mw}[1]{\mathrm{w}_{\mathcal{M}}(#1)}
\newcommand{\gw}[1]{\mathrm{w}_{\mu}(#1)}

\begin{document}

\title{MWS and FWS Codes for Coordinate-Wise Weight Functions \thanks{The first author acknowledges the support of the NSERC of Canada Discovery Grant program.}
}


\author{Tim Alderson        \and
        Benjamin Morine %etc.
}


\institute{T. Alderson \at
              University of New Brunswick \\
              Saint John, NB\\
              Canada
              Tel.: +1506-648-5622\\
              \email{Tim@unb.ca}           %  \\
           \and
           B. Morine \at
              University of New Brunswick\\
              Saint John, NB\\
              Canada
}

\date{Received: date / Accepted: date}

\maketitle

\begin{abstract}
A combinatorial problem concerning the maximum size of the ( Hamming) weight set of an $[n,k]_q$ linear code was recently introduced. Codes attaining the established upper bound are the Maximum Weight Spectrum (MWS) codes. Those $[n,k]_q $ codes with the same weight set as $ \F_q^n $ are called Full Weight Spectrum (FWS) codes. FWS codes are necessarily ``short", whereas MWS codes are necessarily ``long". For fixed $ k,q $ the values of $ n $ for which an $ [n,k]_q $-FWS code exists are completely determined, but the determination of the minimum length $ M(H,k,q) $ of an $ [n,k]_q $-MWS code remains an open problem. The current work broadens discussion first to general coordinate-wise weight functions, and then specifically to the Lee weight and a Manhattan like weight. In the general case we provide bounds on $ n $ for which an FWS code exists, and bounds on $ n $ for which an MWS code exists. When specializing to the Lee  or to the Manhattan setting we are able to completely determine the parameters of FWS codes. As with the  Hamming case, we are able to provide an upper bound on  $ M(\mathcal{L},k,q) $ (the minimum length of Lee MWS codes), and pose the determination of $ M(\mathcal{L},k,q) $ as an open problem. On the other hand, with respect to the Manhattan weight we completely determine the parameters of MWS codes.
\end{abstract}

\section{Introduction}
\label{intro}
In 1973 Delsarte studied the number of distinct distances held by a code $\C$, which in the case of linear codes equates to studying the number of distinct weights \cite{Delsarte1973}. Discussions on the set of distinct weights of a code  can be traced in \cite{MacWilliams1963}, where the author tackled the following question. Given a set of positive integers, $ S $,  is it possible to construct a code whose set of non-zero weights is $ S $? Partial solutions were presented, and  necessary conditions were established. 

In  \cite{Shi2019}  Shi \textit{et. al.} consider a combinatorial problem concerning the maximum number $ L(k,q) $ of distinct non-zero (Hamming) weights  a  linear code over $ \F_q $ with dimension $k$ may possess. Specifically they propose   \begin{equation}
	L(k,q)  \le \frac{q^k-1}{q-1}. \label{eqn: Sole bound}
\end{equation} 
In \cite{Alderson2019}, linear codes meeting this bound were shown to exist for all $ k $ and $ q $, and such codes were named \textit{maximum weight spectrum} (MWS) codes.  

A further refinement was also investigated in  \cite{Shi2019} by the introduction of the function $ L(n, k, q) $, denoting the maximum number of non-zero (Hamming) weights a linear $ [n, k]_q $ code may have. Since an arbitrary vector of length $ n $ has   Hamming weight at most $ n $, an immediate upper bound is  $ L(n, k, q)\le n $. In \cite{MR4014640}, the parameters of linear codes meeting this bound were completely determined, and were named \textit{full weight spectrum} (FWS) codes. 

The bound (\ref{eqn: Sole bound}) was shown to be sharp for binary codes by Haily and Harzalla \cite{Haily2015}. In \cite{Shi2019}, the binary case is proved independently from \cite{Haily2015} and the bound is further shown to be sharp for all $ q $-ary linear codes of dimension $ k=2 $.  The authors went on to conjecture that the bound is sharp for all $ q $ and $ k $. This conjecture was proved correct in \cite{Alderson2019}, where it is shown that linear MWS codes exist for all $ q $ and $ k $. For fixed $ k $ and $ q $ only sufficient conditions on $ n $ are known  for which an MWS code exists  \cite{Alderson2019,Cohen2018,Meneghetti2020} . The questions relating to maximal weight spectra are explored in \cite{Shi2019a} as they relate to cyclic codes, and in \cite{Shi2020} as they relate to quasi-cyclic codes. In both cases existence results are established. The problem of determining the minimum value of $ n $ for which an $ [n,k]_q $ MWS code exists has attracted some attention \cite{Alderson2019,Cohen2018,Meneghetti2020} but remains open.

In this work we expand discussions to linear codes equipped with general coordinate-wise weight functions, with a view towards  Lee weight and a Manhattan type weight. With linear codes, all scalar multiples of any codeword also belongs to the code. Non-zero scalar multiples of a codeword have the same Hamming weight, but in general have differing Lee weights and differing Manhattan weights. Hence, under these new weight functions we might expect bounds regarding respective weight spectra to differ from those established under the Hamming metric. 

In the sequel we establish necessary and sufficient conditions on $ n $ under which FWS codes with respect to these weight functions exist. We are also able to establish the existence of MWS codes for each $ k $ and $ q $, and for the Manhattan weight we are able to determine the minimal length $ n $ for which these MWS codes exist. On the other hand, determining minimal length of an MWS code under the Lee metric remains as elusive as it is in the case of the Hamming metric.  

The research on codes in the Manhattan metric is not extensive. This metric is popular in psychology and related areas but not well known in communications where such channels appear when a signal limiter is used at the receiver \cite{Gabidulin2012}. The literature on codes in the Lee metric is comparatively richer, though not entirely extensive, e.g., \cite{Astola1982a,MR0342262,Etzion2013,Golomb1970,Orlitsky1993,Roth1994,Satyanarayana1979}. Lee-codes were first introduced by C. Y. Lee in \cite{MR0122630}, and were found to be advantageous for certain types of data transmission. For example, the input alphabet may correspond to phases of a sinusoidal signal of fixed amplitude and frequency, to which white Guassian noise is added. Such a channel has an inherent ordering of the letters of the input alphabet, and transitions between adjacent letters are much more likely than transitions between distant input letters. Corrupted digits of phase-modulated signals are more likely to have only a slightly different phase in comparison to the original signal.  For such channels, the Lee metric is a much more useful notion than the Hamming metric   \cite{Berlekamp2015,Berlekamp1969}.

The properties of Lee-codes and, in particular, the question of existence (or nonexistence) of perfect codes in the Lee metric has been studied by numerous authors, for example,  \cite{Araujo2014,Astola1987,Etzion2011,Etzion2010,Golomb1970,Horak2009,Post1975}.   More recent interest in Lee codes has been spurred on due to further development of applications for these codes. Some examples include constrained and partial-response channels \cite{Roth1994}, interleaving schemes  \cite{Blaum1998}, orthogonal frequency-division multiplexing \cite{Schmidt2007}, multidimensional burst-error-correction \cite{Etzion2009},  error-correction in the rank modulation scheme for flash memories \cite{Jiang2010}, and VLSI decoders and fault-tolerant logic \cite{Astola2017,Roth1994}. There have also been more recent attempts to settle the existence question of perfect codes in the Lee metric \cite{Horak2009,Horak2018,Qureshi2020}.

\section{Preliminaries}
\label{sec:1}
Let $ \A $ be an alphabet of size $ q $ which without loss of generality is the set $ \{0,1,\ldots,q-1\} $. An $ (n,M)_q $-code, $ \C $ is  a subset of $ \A^n $ with $ |\C|=M $. The elements of $ \A^n $ are called words or vectors, while elements of $ \C $ are called codewords.  In the case that   $\A= \F_q$ (the finite field with $q$ elements) and $ \C $ is a subspace of dimension $ k $, $ \C $ is a \textit{linear code}, denoted an $ [n,k]_q $-code. We only consider nondegenerate linear codes, which is to say linear codes having no coordinate which is identically zero.  The \textit{Hamming distance}, $ d_H(a,b) $, between two words is the number of coordinates in which they differ.   A \emph{generator matrix} $G$ for an $[n,k]_q$-code  $\C$ is a $k\times n$ matrix over $\F_q$ whose row vectors generate $\C$. The support of a word $ c $ is the set  $ \textrm{supp}(c)=\{i\mid c_i\ne 0\} $.

Throughout we shall adopt standard notation regarding vector spaces. In particular if $ S $ is a set of vectors then $ \langle S\rangle $ shall denote the span of $ S $, by $ e_i $ we shall denote the standard basis vector with $ i $'th entry $ 1 $ and all other entries $ 0 $, and given $ n $  $ k $-vectors $ \lambda_1,\lambda_2,\ldots,\lambda_n $ we shall denote by $ G=[\lambda_1^{\alpha_1}|\lambda_2^{\alpha_2}|,\cdots,|\lambda_n^{\alpha_n}] $ the $ k\times (\alpha_1+\alpha_2+\cdots+\alpha_n) $ matrix obtained by repeating the $ i $'th column $ \alpha_i $ times, $1\le i \le n$.   

%\subsection{Subsection title}

\subsection{Component-wise Metrics and Weight Functions, MWS and FWS Codes}
\label{subsec:2}
For our purposes a weight function on $ \F_q^n $ is a mapping $ w_\mu:\F_q^n\to {\mathbb R}^+ $ with $ \gw{x} =0 $ if and only if $ x=0 $;  $ \gw{u} $ is referred 
to as the $ \mu $-weight of $ u $. A weight function $ w_\mu$ on $\F_q$ may be extended in a natural way (and with a slight abuse of notation) to a weight function on $\F_q^n$ for each $n\ge 1$  by taking $ \gw{u}= \gw{u_1} + \gw{u_2}+\cdots + \gw{u_n} $. In such a case we say that $ w_\mu $ is a \textit{component-wise weight function}. For any subset $ \C\subseteq \F_q^n $ the $ \mu- $\emph{weight set of $\C$} is defined as  $$ \gw{\C}=\left\{ \gw{c}   \mid c \in \C\setminus \{0\}\right\}.$$

For $ u\in \F_q^n $, the Hamming weight of $ u $ is given by $ w_H(u)=d_H(u,0) $,  the (Hamming) distance between $ u $  and $ 0 $. Likewise, given any metric on $ \F_q^n $ one has, \textit{mutatis mutandis} an associated weight function.

If $ d_\mu :\A^n\times \A^n\to \Z^+$ is a metric then we denote by $ d_\mu (u,v) $ the $ \mu $-distance between $ u $ and $ v $, and similarly $  \gw{u} $ denotes the $ \mu $-weight (distance from $ 0 $) of $ u $.  By $ S_\mu(u,t) $ we denote the $ \mu $-sphere of radius $ t $ around $ u $, and by $ \sh_\mu(u,t) $ it's boundary, that is
\[
S_\mu(u,t) = \{ x \mid d_\mu(u,x) \le t\},\text{ and } \sh_\mu(u,t) = \{ x \mid d_\mu(u,x) = t\}.  
\] 
We say $ d_\mu $  is a component-wise metric if $ d_\mu(u,v) = \sum\limits_{i=1}^nd_\mu(u_i,v_i)$, and note that $ d_\mu $ then corresponds to a component-wise weight function $ w_\mu $. The Hamming metric is an example of component-wise metric, as are both the Lee and Manhattan metrics.\\

Analogous to the Hamming case, for fixed $ k,q $ we denote by $ L_\mu(k,q) $ the maximum cardinality of $  \gw{\C}  $ as $ \C $ ranges over all $ q $-ary linear codes of dimension $ k $. If $ n $ is also fixed then this quantity is denoted by $ L_\mu(n,k,q) $. Since an $ [n,k]_q $ code $ \C $ is a subset of $ \F_q^n $ it follows that $ |w_\mu(\C)|\le |w_\mu(\F_q^n)| $ motivating the following definition.   

\begin{definition}\label{defn: FWS codes general}
	If $ \C $ is a linear $ [n,k]_q $-code then $ \C $	is a $ \mu $-Full Weight Spectrum ($ \mu $-FWS) code if $ w_\mu(\C)= w_\mu(\F_q^n) $.
\end{definition}
Note that $ |w_H(\F_q^n)| =n  $, so the definition above subsumes that of FWS codes with respect to Hamming weights. Clearly, if $ n=k $ then $ \C $ is necessarily FWS, so the task at hand is that of determining the maximum length $ n $ of a $ k $-dimensional $ q $-ary $ \mu $-FWS code.  The following inequality also follows from the respective definitions:
\begin{equation}\label{eqn: compare L(kq) and L(nkq)}
	L_\mu(n,k,q)\le L_\mu(k,q).
\end{equation} 
An open problem posed in  \cite{Shi2019} is that of determining the sharpness of (\ref{eqn: compare L(kq) and L(nkq)}) for the Hamming weight; to determine parameters for which the inequality is strict. This is also of interest when widening discussions to other types of weight functions. \\

For a component-wise weight function on $ \F_q^n$, a further parameter shall prove useful. 

\begin{definition}\label{defn: delta}
If $ w_\mu $ is a component-wise weight function  on $\F_q^n$ then we define the parameter $\Delta_\mu$ by   \begin{equation}
	\Delta_\mu =\max\limits_{n\ge k} \{|w_\mu( \langle u \rangle)|\;\mid u\in \F_q^n\}.	
\end{equation}    
For notational convenience, if $w_\mu$  is clear from the context then we write $\Delta_\mu=\Delta$.
\end{definition}

The following is easily established.

\begin{proposition}\label{prop: Maximum Weight Spectrum and Full Weight Spectrum Codes}
	Let $ w_\mu $ be a component-wise weight function  on $ \A^n$ where $ |\A|=q $, and let   $ m=\max w_\mu(\A)  $.
	\begin{enumerate}
		\item  $ |w_\mu(\A^n)| \le n\cdot m $. \label{propP1: Maximum Weight Spectrum and Full Weight Spectrum Codes}
		\item If $ w_\mu(\A)=\{1,2,\ldots,m\} $ then $ |w_\mu(\A^n)| = n\cdot m $.\label{propP2: Maximum Weight Spectrum and Full Weight Spectrum Codes}
		\item If $ w_\mu(\A)=\{1,2,\ldots,m\} $ then an $ (n,M)_q $ code  $ \C $ over $ \A $ is $ \mu $-FWS if and only if $ |w_\mu(\C)| = n\cdot m $.\label{propP2b: Maximum Weight Spectrum and Full Weight Spectrum Codes}
		\item If $ \A=\F_q $  then    $ L_\mu(k,q) \le \frac{q^k-1}{q-1}\cdot \Delta $ .   
	\end{enumerate}  
\end{proposition} 

We may now generalize the notion of $MWS$ to the more general setting of component-wise weight function.

\begin{definition}\label{defn: linear MWS codes general}
	Let $ w_\mu $ be a component-wise weight function  on $ \F_q^n$.  A linear $ [n,k]_q $-code $ \C $	is a $ \mu $-Maximum Weight Spectrum ($ \mu $-MWS) code if   \[|w_\mu(\C)|= \frac{q^k-1}{q-1}\cdot \Delta .\]
\end{definition}

\begin{remark} A few observations regarding the Proposition \ref{prop: Maximum Weight Spectrum and Full Weight Spectrum Codes} and Definition \ref{defn: linear MWS codes general}:
	\begin{enumerate}\item If $ w_\mu(\A)=\{1,2,\ldots,m\} $ then we may characterize $ \mu $-FWS codes as those having at least one codeword of each $ \mu $-weight from $ 1 $ to $ mn $.
		\item We may characterize linear $ \mu $-MWS codes  as those linear codes in which every one dimensional subspace has a $ \mu $-weight set of size $ \Delta $, and the $ \mu $-weight sets of one dimensional subspaces are mutually disjoint.     
		\item  For the Hamming weight, $ m=1=\Delta $, and as such the terminology is compatible with that of \cite{Shi2019}.	
	\end{enumerate}
\end{remark}

%In all of the above notations we may drop the subscript $ \mu $ if the metric/weight of concern is clear from the context. 
In the sequel we will find the following to be of use.

\begin{lemma}\label{lem: bounds on length in general case} 
	Let $ \C $ be an $  [n,k ]_q $-code, and let $ w_\mu $ be a component-wise weight function defined on $ \F_q^n $ with $ w_\mu(\F_q)=\{1,2,\ldots,m\} $.
	\begin{enumerate}
		\item  If   $ \C $ is $ \mu $-FWS then $ n\cdot m\le L_\mu(k,q) $. \label{Part 2: lemma on bounds in general case}
		\item  If   $ \C $ is both $ \mu $-FWS and $ \mu $-MWS, then $ n=\frac{q^k-1}{q-1}\cdot \frac{\Delta}{m}$. \label{Part 4: lemma on bounds in general case}
	\end{enumerate} 	
\end{lemma}
\begin{proof}
	The first part follows from the Eq. (\ref{eqn: compare L(kq) and L(nkq)}) and Proposition \ref{prop: Maximum Weight Spectrum and Full Weight Spectrum Codes}. For the second part, observe that the FWS property gives $ n=\frac{|w_\mu(\C)|}{m}$, whereas the MWS property gives $ |w_\mu(\C)|=\frac{q^k-1}{q-1}\cdot \Delta $.  
\end{proof}

\begin{lemma}\label{lem: more bounds on length in general case} 
	Let $ \C $ be an $ (n,M)_q $-code containing the zero codeword, let $ w_\mu $ be a component-wise weight function,  and let $ r = \max\limits_{\delta}  |\{u\in \C\mid w_\mu(u)=\delta\}| $.   It holds that \[\frac{M-1}{r}\le |w_\mu(\C)| \le  m\cdot n,\] with equality throughout if and only if $ |\{u\in \C\mid w_\mu(u)=\delta\}|= r $ for all $ 1\le \delta\le mn $. 
\end{lemma}
\begin{proof}
	The first inequality arises by counting ordered pairs $ (u,w(u)) $ where $ 0\ne u\in \C $ in two ways to give
	$ M-1\le |w_\mu(\C)| \cdot r $, and the second inequality is from Proposition \ref{prop: Maximum Weight Spectrum and Full Weight Spectrum Codes}.  	The equality part follows from simple counting.   
\end{proof}

Given the lower and upper bounds inherent in the above it behooves us to investigate optimality with respect to code length. To this end we define $ M(\mu, k, q) $ and $ N(\mu, k, q) $ as follows:
\begin{equation}
	M(\mu, k, q) = \min \{n \mid C \text{ is an } [n,k]_q \;\; \mu-MWS \text{ code } \}
\end{equation}  
\begin{equation}
	N(\mu, k, q) = \max \{n \mid C \text{ is an } [n,k]_q \;\; \mu-FWS \text{ code } \}
\end{equation}  

Tables 1 - \ref{tab:3} will perhaps place some of the results that follow in context.\\
\begin{table}
\parbox{.5\linewidth}{
%\begin{table}
	\caption{Bounds: Hamming Weight}
	\label{tab:1}       % Give a unique label
	\begin{tabular}{ll}
		\hline\noalign{\smallskip}
		Hamming Weight	&  Reference \\
		\noalign{\smallskip}\hline\noalign{\smallskip}
		$ M(H,k,2)= 2^k-1 $ 	&   \cite{MR4014640,Shi2019} \\
	\hline
	$ M(H,3,3)\le 32 $	& \cite{Alderson2019} \\
	\hline
	$  \frac{q}{2}\cdot  \frac{q^k-1}{q-1} \le  M(H,k,q) $   	& \cite{MR4014640}\\\hline
	$    M(H,k,q)\le \frac{q}{2}\cdot  \frac{q^k-1}{q-1}\cdot   \frac{q^{k-1}-1}{q-1} $   	& \cite{Alderson2019}\\
	\hline
	$ 	M(H,2,q)=\frac{q(q+1)}{2} $ & \cite{MR4014640} \\
	\hline
	$ N(H,k,q) = 2^k-1$	& \cite{MR4014640} \\
		\noalign{\smallskip}\hline
	\end{tabular}
%\end{table}
}\hfill
\parbox{.45\linewidth}{
%\begin{table}
	\caption{Bounds: Manhattan Weight}
	\label{tab:2}       % Give a unique label
	\begin{tabular}{ll}
		\hline\noalign{\smallskip}
		Manhattan Weight	&  Reference \\
		\noalign{\smallskip}\hline\noalign{\smallskip}
			$  M(\MM,k,q)= \frac{ q^k-1}{q-1} $   	& Thm. \ref{thm: MMWS bound determined} \\
		\hline
		$ N(\MM,k,q)= \frac{ q^k-1}{q-1} $	& Thm. \ref{thm: ub length of MFWS}  \\
		\noalign{\smallskip}\hline
	\end{tabular}
}
\end{table}
\begin{table}
	\caption{Bounds: Lee Weight}
	\label{tab:3}       % Give a unique label
	\begin{tabular}{ll}
		\hline\noalign{\smallskip}
			Lee Weight	&  Reference \\
		\noalign{\smallskip}\hline\noalign{\smallskip}
		$   \frac{q^k-1}{q-1} +\left\lceil \frac{2(k-1)}{q-1}\right\rceil \le  M(\LL,k,q>2) $ 	& Lem. \ref{lem: lb length of LMWS}\\ 
	\hline
	$   2^k-1  \le  M(\LL,k,2) $   	& Lem. \ref{lem: lb length of LMWS} \\
	\hline
	$ M(\LL, k, q)\le \frac{ (\frac{q+1}{2})^{2k^2-3k+1}-1}{\frac{q+1}{2}-1} $ &  Thm. \ref{thm: Lee MWS bound determined}\\ 
	\hline
	$ M(\LL,2,7)\le 16 $	& \cite{BenMorineMScThesis2022} \\
	\hline
	$ M(\LL,2,11)\le 34 $   	&  \cite{BenMorineMScThesis2022}\\
	\hline
	$ M(\LL,2,13)\le 46 $ & \cite{BenMorineMScThesis2022} \\
	\hline
	$ M(\LL,2,17)\le 76 $	&  \cite{BenMorineMScThesis2022}\\
	\hline
	$ M(\LL,2,19)\le 86 $	& \cite{BenMorineMScThesis2022} \\
	\hline
	$ M(\LL,2,23)\le 126 $	&\cite{BenMorineMScThesis2022}  \\
	\hline
	$ N(\LL,k,q)= \frac{\left(\frac{q+1}{2}\right)^k-1}{\left(\frac{q+1}{2}\right)-1} $	& Lem. \ref{lem: LFWS necessary and sufficient}  \\
		\noalign{\smallskip}\hline
	\end{tabular}
\end{table}

In the next section we focus on determining $ N(\mu, k, q) $.

\section{Generalized Full Weight Spectrum Codes}
\label{sec: GFWS}
We can in fact determine $ N(\mu, k, q) $ when $ w_\mu(\F_q)=\{ 1, 2, \ldots,m\} $. To do this, the following will be of use.

\begin{lemma}\label{lem: General FWS basis property}
	Suppose $ w_\mu $ is a coordinate-wise weight function on $ \F_q^n $, where $ \F_q=\{ 0, \alpha_1=1,\alpha_2,\ldots,\alpha_{q-1}\} $, and $ w_\mu(\F_q)=\{ 1, 2, \ldots,m\} $.\\ 
	Let $ \C $ be a (linear) $ [n,k]_q $ code with basis $B=\{v_1,v_2,\ldots,v_k\}$. Let $ S= \textrm{supp}(v_1) $, $ \displaystyle T=\cup_{i=2}^{k} \textrm{supp}(v_i) $,  $ S'=S\setminus T $, with $ s=|S|, t=|T| $, and $ s'=|S'| $.\\ If $ C $ is $ \mu $-FWS then
	\[
	s'\le m\cdot t +1.
	\]  	
\end{lemma}

%\newpage 

\begin{proof}
	If $ s'=0 $ then the result holds, so assume $ s'\ge 1$.
	Let $ x\in \C $. If $ \textrm{supp}(x)\cap S'=\emptyset $ then $ \textrm{supp}(x)\subseteq T $, so $ |\textrm{supp}(x)|\le t $ giving   
	\begin{equation}\label{eqn:basis general FWS 1} 
		w_\mu(x)\le m\cdot t.
	\end{equation} 	
	If $ \textrm{supp}(x)\cap S'\ne \emptyset $ then it follows that  $ S' \subseteq \textrm{supp}(x)   $, so $ |\textrm{supp}(x)|\ge s' $ giving   
	\begin{equation}\label{eqn:basis general FWS 2} 
		w_\mu(x)\ge  s'.
	\end{equation}  
	Since $ n=t+s' $ (by assumption no coordinate is identically zero) we see that if  $ \C $ is $ \mu $-FWS then 
	\[
	w_\mu(\C) = \left\{1,2,\ldots,(t+s')\cdot m \right\},
	\]
	and since $ m,s'\ge 1  $ there exists a codeword $ u $ with $ w_\mu(u)= t \cdot m +1$. From  (\ref{eqn:basis general FWS 1}) and (\ref{eqn:basis general FWS 2}) we may conclude that
	\[
	w_\mu(u)= t \cdot m+1 \ge  s'
	\]
	giving the desired inequality.   
\end{proof}

\begin{theorem}\label{thm: ub2 length of General FWS}
	Let $ m $ be as in Lemma \ref{lem: General FWS basis property}. If $ \C $ is a linear $ [n,k]_q $ $ \mu $-FWS code then 
	\[
	n\le \sum\limits_{i=0}^{k-1} \left(m+1\right)^i =\frac{(m+1)^k-1}{m}.
	\]
\end{theorem}
\begin{proof}
	Let $ \C $ be an $ [n,k]_q $ $ \mu $-FWS code. It suffices to show that we may select a basis $ \{\lambda_1,\lambda_2,\ldots,\lambda_k\} $ of $ \C $ such that
	\begin{equation}\label{eqn: slecting bases general FWS}
		\left|\bigcup_{i=1}^t \textrm{supp}(\lambda_i) \right|\le 1 + (m+1)+\cdots +(m+1)^{t-1}, \text{ for } 1\le t \le k. 
	\end{equation}
	To this end first note that since $ \C $ is $ \mu $-FWS, there is a codeword of weight $ 1 $. Assume without loss of generality that $ w_\mu(1)=1 $, so  we may take $ \lambda_1=e_1 $, so (\ref{eqn: slecting bases general FWS}) holds for $ t=1 $. Assume $ \lambda_1,\lambda_2,\ldots,\lambda_{k-1} $ have been selected in accordance with (\ref{eqn: slecting bases general FWS}), and let $ \lambda_k $ be a codeword such that $ \{\lambda_1,\lambda_2,\ldots,\lambda_k\} $ is a basis for $ \C $.  From the Lemma \ref{lem: General FWS basis property} we have
	\begin{align*}
		\left|\bigcup_{i=1}^k \textrm{supp}(\lambda_i) \right|& \le \left|\bigcup_{i=1}^{k-1}\textrm{supp}(\lambda_i)\right|+ m\cdot \left|\bigcup_{i=1}^{k-1}\textrm{supp}(\lambda_i)\right|+1\\
		& = (m+1)\cdot \left|\bigcup_{i=1}^{k-1}\textrm{supp}(\lambda_i)\right|+1\\
		& \le \sum\limits_{i=0}^{k-1} \left(m+1\right)^i.
	\end{align*}
\end{proof}

We now show the bound in Theorem \ref{thm: ub2 length of General FWS} is sharp.

\begin{theorem}\label{thm: General FWS necessary and sufficient}
	Let $ m $ be as in Lemma \ref{lem: General FWS basis property}. There exists an $ [n,k]_q $ $ \mu $-FWS code if and only if   \[k\le n \le \frac{(m+1)^k-1}{m}.\]
	In particular $ N(\mu,k,q)= \frac{(m+1)^k-1}{m}.$
\end{theorem}
\begin{proof}
	That $ k\le n \le \frac{(m+1)^k-1}{m} $ is necessary is given in Theorem \ref{thm: ub2 length of General FWS}.	
	For sufficiency, let $ k $ be fixed and let $ G $ be the following matrix:
	\begin{equation*}
		G=\left[ \begin{pmatrix}
			1\\0\\0\\ \vdots \\ \vdots \\ 0 
		\end{pmatrix}  
		\begin{pmatrix}
			0\\1\\0\\ \vdots \\ \vdots \\ 0 
		\end{pmatrix} 
		\begin{pmatrix}
			0\\1\\0\\ \vdots \\ \vdots \\ 0 
		\end{pmatrix}
		\cdots
		\begin{pmatrix}
			0\\1\\0\\ \vdots \\ \vdots \\ 0 
		\end{pmatrix}
		\begin{pmatrix}
			0\\0\\1\\ \vdots \\ \vdots \\ 0 
		\end{pmatrix}
		\begin{pmatrix}
			0\\0\\1\\ \vdots \\ \vdots \\ 0 
		\end{pmatrix} \cdots
		\begin{pmatrix}
			0\\0\\1\\ \vdots \\ \vdots \\ 0 
		\end{pmatrix}
		\cdots \cdots
		\begin{pmatrix}
			0\\0\\0\\ \vdots \\ \vdots \\ 1 
		\end{pmatrix}
		\right]
	\end{equation*} 
	\[
	=\left[ e_1\left|(e_2)^{ m+1 }\left|(e_3)^{( m+1 )^2}\right|\cdots \right|(e_k)^{( m+1 )^{k-1}}\right]
	\]
	Denote the rows of $ G $ by $ \lambda_1,\lambda_2,\ldots \lambda_k $ and let $ \C $ be the code generated by $ G $. The code $\C$ has length $n= \frac{(m+1)^k-1}{m} $. Moreover, for any $ u=(\alpha_1,\alpha_2,\ldots,\alpha_k) \in \F_q^k $,  

	\begin{align*}
		w_{\mu}(uG) =w_{\mu}(\alpha_1\lambda_1+\cdots+\alpha_k\lambda_k) & = \gw{\alpha_1}+\gw{\alpha_2}\left( m+1 \right) +\cdots+\gw{\alpha_k}\left( m+1 \right)^{k-1}\\ & =\sum_{i=1}^{k}\gw{\alpha_i}\left( m+1 \right)^{i-1}.
	\end{align*}  
	Since $ \gw{\F_q}=\{1,2,\ldots,m\} $ we obtain  $ |w_{\mu}(\C)|=nm (=(m+1)^k-1)  $, so $ \C $ is an FWS.\\
	More generally, for $k\le n \le \frac{(m+1)^k-1}{m}$ consider the $k\times n$ generator matrix $ G' $ obtained by  removing rightmost columns of $ G $ while preserving rank, say
	\[
	G'= \left.\left.\left.\left[ e_1\left|(e_2)^{ m+1 }\left|%(e_3)^{( m+1 )^2}\right|
	\cdots \right|(e_{t})^{( m+1 )^{t-1}}\left|(e_{t+1})^{\alpha}\right| e_{t+2}\right|e_{t+3}\right|\cdots \right|e_{k} \right. \right],
	\]
	where $1\le \alpha \le (m+1)^{t}$. 	Let $ \C' $ denote the code with generator $ G' $. 
	
	Let $\gamma$ be an arbitrary $\mu-$weight with $1\le \gamma \le mn$.  If $\gamma \le m(k-t-1)$ then $\gamma = \gw{ uG'}$ for any $u = (u_1,u_2,\ldots,u_k)$ with final $k-t-1$ coordinates chosen such that the sum of their $\mu-$weights is $\gamma$, the remaining coordinates being $0$. Now consider $\gamma > m(k-t-1)$ and let $a = \left\lfloor \frac{\gamma - m(k-t-1)}{\alpha}\right\rfloor$. If $a\le m$ then $b=\gamma - a\alpha -m(k-t-1)\le (m+1)^t-1$, so we may expand $b$ in base $m+1$:
	\[
	b=b_1+b_2(m+1)+b_3(m+1)^2+\cdots +b_t(m+1)^{t-1}.
	\]    
	Taking $u\in \F_q^k$ with $\gw{u_i}=b_i$, $1\le i \le t$, $\gw{u_{t+1}}=a$, and $\gw{u_i}= m$ for $i>t+1$ we obtain $\gw{uG'} = \gamma$. Likewise, if $a>m$ we may take $u\in \F_q^k$ with $\gw{u_i}=b_i$, $1\le i \le t$, and $\gw{u_i}= m$ for $i>t$ to obtain a codeword of weight $\gamma$.
\end{proof}

\begin{corollary} \label{cor: FWS and MWS}
	Let $ m $ be as in Proposition \ref{prop: Maximum Weight Spectrum and Full Weight Spectrum Codes}. If $ \C $ is an $ [n,k]_q $ code that is both $ \mu $-MWS and $ \mu $-FWS, then $ \Delta=m=q-1 $, and in particular $ n=\frac{q^k-1}{q-1} $ 
\end{corollary}
\begin{proof}
	First observe that since $ w_\mu $ is a component-wise weight function, $ w_\mu(\F_q)=w_\mu(\left\langle e_1 \right\rangle) $, giving $ m\le \Delta $.  Next we observe that
	\begin{equation} \label{eqn: FWS and MWS}
		n=\frac{q^k-1}{q-1}\cdot \frac{\Delta}{m} \le  \frac{(m+1)^k-1}{m}\le \frac{q^k-1}{q-1}.
	\end{equation}
	The first equality follows from Lemma \ref{lem: bounds on length in general case} , the first inequality follows from Theorem \ref{thm: General FWS necessary and sufficient}, and the final inequality holds since $ m\le q-1 $. From (\ref{eqn: FWS and MWS}) it follows that $ \Delta \le m $ whence $ \Delta = m $ and the result follows.
\end{proof}

Corollary \ref{cor: FWS and MWS} tells us that if we seek codes that are both MWS and FWS then the particular weight function must have a rather strict structure.  Via a Manhattan type weight function we will exhibit infinite families of codes that are both MWS and FWS.

\section{Codes with the Lee and Manhattan Weight Functions}
\label{sec: LMWF}

In this section we specialize discussions to the Lee and Manhattan weight functions. We introduce these weight functions in the more general setting of $\Z_q^n$, where $\mathbb{Z}_{q}=\{0,1,\ldots,q-1\}$ denotes the ring of integer residues modulo the positive integer $q$. Linear codes are defined over fields, so when applying these weights to linear codes it is necessarily assumed that $q$ is prime.   

Roughly following the notation and terminology as established in \cite{1137784}, the Lee weight of an element $\alpha \in \mathbb{Z}_{q}$, denoted by $w_{\mathcal{L}}(\alpha)$ or $|\alpha|$, takes non-negative integer values and is defined by

\begin{equation}
	w_{\mathcal{L}}(\alpha)=|\alpha|= \begin{cases} \alpha  & \text { if } 0 \leq \alpha  \leq q / 2 \\ q- \alpha  & \text { otherwise. }\end{cases}
\end{equation}

We refer to the elements $1,2, \ldots,\lfloor q / 2\rfloor$ as the ``positive" elements of $\mathbb{Z}_{q}$ and denote the set of such elements $ \Z_q^+ $ ; the remaining elements in $\Z_{q} \setminus \{0\}$ compose the ``negative" elements, denoted $\Z_{q}^{-}$.

\begin{example}
	In $\Z_{7}$,  $\Z_{7}^{+}=\{1,2,3\}$ and  $ |2|=|5|=2 $, whereas in $\Z_{8}$, $\Z_{8}^{+}=\{1,2,3,4\}$ and $|2|=|6|=2$.	
\end{example}

For a word $ {c}=\left(c_{1}, c_{2}, \ldots, c_{n}\right)$ in $\mathbb{Z}_{q}^{n}$, we define the Lee weight of $  {c} $ by
$$
w_{\mathcal{L}}( {c})=\sum_{j=1}^{n}\left|c_{j}\right|.
$$

(with the summation being integer summation). The Lee distance (metric) between two words $ {x},  {y} \in \Z_{q}^{n}$ is defined as $ w_\mathcal{L}( {x}- {y})$; we denote that distance by $\mathrm{d}_{\mathcal{L}}( {x},  {y})$. We note that for $ q=2,3 $ the Hamming and Lee metrics coincide. 

In keeping with the notation introduced in the previous section we shall adopt the following:

\begin{equation}
	L_{\mathcal{L}}(k,q) :=\max \{|w_\mathcal{L}(\mathcal{C})|: \mathcal{C} \text { is linear, of dimension } k \text{ over } \Z_q \}, 
\end{equation} 
and 
\begin{equation}
	L_{\mathcal{L}}(n,k,q) :=\max \{|w_\mathcal{L}(\mathcal{C})|: \mathcal{C} \text { is an } [n,k]_q  \text{ code over }  \Z_q \}. 
\end{equation}

The minimum Lee distance of an $(n, M)$ code $\mathcal{C}$ over $\mathbb{Z}_{q}$ with $M>1$ is defined by
$$
\mathrm{d}_{\mathcal{L}}(\mathcal{C})=\min _{ {c}_{1},  {c}_{2} \in \mathcal{C}:  {c}_{1} \neq  {c}_{2}} \mathrm{~d}_{\mathcal{L}}\left( {c}_{1},  {c}_{2}\right) .
$$
We note that for a linear code, the minimum distance is also the minimum weight, that is    \[\mathrm{d}_{\mathcal{L}}(\mathcal{\C}) = \min w_{\LL}(\C).\]

	The Manhattan metric (or Taxicab metric), is defined for alphabet letters taken from the integers. For two words $ {{  x}=(x_{1},x_{2},\ldots, x_{n})}, {  y}=(y_{1},y_{2},\ldots, y_{n}) \in   \mathbb Z^n $ the Manhattan distance between $ {  x} $ and $ y $ is defined as  \begin{equation}
		d_{\mathcal{M}}({  x},{  y}) = \sum\limits_{i=1}^{n} |x_i-y_i|.
	\end{equation} The Manhattan weight of a vector in $ \mathbb Z^n $ is it's distance from zero. As we have chosen $\mathbb{Z}_{q}=\{0,1,\ldots,q-1\}$ we may adapt the Manhattan weight to $ \Z_q^n $ by treating coordinate entries as integers. Doing so we define the Manhattan weight of a vector $ x\in \Z_q^n $ as the integer sum: 
	\begin{equation}
		\mw{ {x}} =\sum_{i=1}^n x_i.
	\end{equation}

%	The Lee and Manhattan weights are defined to work over the ring $ \Z_q $. If the size of the alphabet is not prime then $ \Z_q $ is not a field. As our focus here shall be on linear codes we shall (unless otherwise stated) discuss only codes over fields $ \Z_q $, where $ q=p $ is prime. 

	\subsection{Lee MWS Codes}
\label{subsec: LMWS}	
	Let $ \C $  be a linear $ [n,k]_q $ code over $ \Z_q $. From the definition of Lee weight it follows that if $  {x} \in \Z_q^n  $ then $ \lw{ {x}} = \lw{ {-x}}$ (i.e. ,$ | {x}|=| {-x}| $). With reference to Definition \ref{defn: delta}, we therefore have $ \Delta \le \frac{q-1}{2} $. One may quickly verify that since $ q $ is prime,  $ |w_\LL(\langle e_1 \rangle)|=\frac{q-1}{2} $ whence $ \Delta=\frac{q-1}{2} $ giving the following.
	
	\begin{theorem}\label{thm: defn for MWS Lee}
		A linear $ [n,k]_q $ code is $ \LL $-MWS if and only if  $ w_\LL(\C) = \frac{q^k-1}{2} $.
	\end{theorem}   
	
	From the definition it is clear that if $ \C $ is  $ \mathcal L $-MWS then two codewords $ u,v $ have the same $ \LL $-weight if and only if $ u=\pm v $.  Hence, in Lemma \ref{lem: more bounds on length in general case} we have $ m=\frac{q-1}{2}, r = 2, $ and $ M=q^k $ giving
	\begin{equation}\label{eqn: low lb on length of LMWS}
		\C   \text{ an } [n,k]_q \;\LL\text{-MWS code } \implies 	n\ge \frac{q^k-1}{q-1}.
	\end{equation}
	We can make a slight improvement to the lower bound (\ref{eqn: low lb on length of LMWS}).  
	\begin{lemma}\label{lem: disjoint support}
		If $ \C $ is an $ [ n,k ]_q $ $ \mathcal L $-MWS code, $ q>2 $, then any two non-zero codewords have at least one support position in common.
	\end{lemma} 
	\begin{proof}
		If $ u,v \in \C$ have disjoint supports then   \[\lw{u+v}=\lw{u}+ \lw{v}=\lw{u-v}=\lw{-u+v}=\lw{-u-v}.\]   Since $ q>2 $ we have $ 1\ne -1 $, and thus by the $ \LL $-MWS property we must have $ u=0 $ or $ v=0 $.
	\end{proof}

	\begin{corollary} \label{cor: supp at least k}
		If $ \C $ is an $  [n,k]_q $ $ \LL $-MWS code, $ q>2 $, and	$ 0\ne c\in \C $,  then $ |\textrm{supp}(c)|\ge k $.  In particular, $ d_{\LL}(\C)\ge k $.
	\end{corollary} 
	\begin{proof}
		Let $ c $ be a codeword with minimal support in $ \C $, say $ |\textrm{supp}(c)|=\delta $. Since $ \C $ is a linear code, in any fixed $ t\le k $ positions there are at least $ q^{k-t} $ codewords with $ 0 $ in each of these coordinate position.  If $ \delta < k $ then Lemma \ref{lem: disjoint support} gives a contradiction.  	
	\end{proof}
	
	\begin{remark}
In the language of \cite{Cohen1985}, Lemma (\ref{lem: disjoint support}) shows that in the non-binary cases,  $ \mathcal L $-MWS codes are \textit{intersection codes}.   Corollary \ref{cor: supp at least k} thus follows from Corollary 1 of \cite{Lempel1977} with the observation that $ d_{\LL}\ge d_H$.  		
	\end{remark}	
	
	\begin{lemma} \label{lem: lb length of LMWS}
		If $ \C $ is an  $  [n,k]_q $ $ \LL- $MWS code with $ q>2 $ then $ n\ge \frac{q^k-1}{q-1} + \left\lceil\frac{2(k-1)}{q-1}\right\rceil$. An   $  [n,k]_2 $ $ \LL- $MWS code satisfies $ n\ge 2^k-1$.  
	\end{lemma}
	\begin{proof}
		The case $ q=2 $ follows from (\ref{eqn: low lb on length of LMWS}).  Assume $ q>2 $. The number of possible distinct non-zero Lee weights of a vector $ c $ of length $ n $ is $ \left\lceil\frac{q-1}{2}\right\rceil\cdot n $, moreover, if $ c\in \C $ then we may reduce this number by $ k-1 $ (Cor. \ref{cor: supp at least k}) and the $ \LL- $MWS property gives
		\[
		\frac{q-1}{2} \cdot n -k+1 \ge \frac{q^k-1}{2}.
		\]
		
	\end{proof}
	
	\subsection{Existence of $ \LL $-MWS codes for small $ k,q $}\label{subsec: Existence of LMWS codes for small k,q}
	For small values of $ k $ and $ q $ we now provide some values of  $ L(n,k,q) $ and $ L(k,q) $. For dimension $ k =1$   we simply observe that for all prime $ q $, $w_{\LL}( \langle e_1\rangle)=\frac{q-1}{2}=L(1,q) $ and we thus obtain (rather trivial) MWS codes for any $ q $. For $ q=2 $, and $ q=3 $ the Hamming and Lee metrics coincide, so  \cite{Alderson2019} and \cite{Shi2019} do provide some results. For small values of $ q $ and $ k $, computations were carried out by Morine in \cite{BenMorineMScThesis2022}, see Table \ref{table: optimap bounds for Lee Metric}. 
	
	\begin{table}[h]
		\caption{Optimal bounds for the Lee Metric, small values of $ q $ prime, $ n\ge 1 $}	 \label{table: optimap bounds for Lee Metric}
		\begin{center}
			\begin{tabular}{|c|c|c|c|}
				\hline
				Bound	& Reference & Bound	& Reference \\
				\hline
				$ L_{\LL}(n,1,q) =\frac{q-1}{2} $ ($ \star $)	& See above  &		$ L_{\LL}(k,2)=2^k-1 $ 	&  \cite{Shi2019} \\
				\hline
				$ L_{\LL}(32,3,3)=13 $	$ (\star)  $&  \cite{Alderson2019} &  	 $ L_{\LL}(2,3)=6 $ & \cite{Alderson2019} \\
				\hline
				$ L_{\LL}(16,2,7)=   24 $ $ (\star)  $& \cite{BenMorineMScThesis2022}  &  	&  \\
				\hline
				$ L_{\LL}(46,2,13)=   84  $ $ (\star)  $	&  \cite{BenMorineMScThesis2022}  &	$ L_{\LL}(2,2,5)= 4 $	&  \cite{BenMorineMScThesis2022}\\
				\hline
				$ L_{\LL}(86,2,19)=   180 $	 $ (\star)  $	& \cite{BenMorineMScThesis2022}  & $ L_{\LL}(3,2,5)= 6 $	& \cite{BenMorineMScThesis2022} \\
				\hline
				$ L_{\LL}(11,2,5)=    12 $ $ (\star)  $		& \cite{BenMorineMScThesis2022}		& $ L_{\LL}(4,2,5)= 8 $ & \cite{BenMorineMScThesis2022} \\
				\hline
				$ L_{\LL}(34,2,11)=  60  $ $ (\star)  $ & \cite{BenMorineMScThesis2022}	&  $ L_{\LL}(5,2,5)= 8 $	& \cite{BenMorineMScThesis2022} \\
				\hline
				$ L_{\LL}(76,2,17)=  144  $ $ (\star)  $ & \cite{BenMorineMScThesis2022}  &	$ L_{\LL}(6,2,5)= 9 $		& \cite{BenMorineMScThesis2022} \\
				\hline
				$ L_{\LL}(126,2,23) =  264 $ $ (\star)  $	& \cite{BenMorineMScThesis2022} &  $ L_{\LL}(7,2,5)= 9 $ 	& \cite{BenMorineMScThesis2022} \\
				\hline
			\end{tabular}
		\end{center}
	\end{table}

		Entries marked with a $ \star $  indicate codes that are MWS.
		In the table, the bounds in \cite{BenMorineMScThesis2022} were found using exhaustive computer search. Note that based on Table \ref{table: optimap bounds for Lee Metric} we have
		\[
		8\le M(\LL, 2,5)\le 11.
		\]

	%	\newpage

		\subsection{Existence of $ \LL $-MWS codes} \label{subsec: e-LMWS}
		Having seen some examples of  $ \LL $-MWS codes for specific values of $ q $ and $ k $,  we now provide a general construction for $ \LL $-MWS codes, establishing their existence for all $ q, k $. We shall find the following to be of use in our construction.

		\begin{lemma}\label{lem: absolute values giving equality}
			Let $ u_i,v_i \in \Z_r $, $ 1\le i \le k $, where $ r>2 $ is an odd integer. 
			\begin{enumerate}
				\item If $ |u_1-u_2|=|u_1+u_2| $ then $ u_1=0$ or  $u_2=0 $.
				\item If $ |u_i|=|v_i| $, and $ |u_i+u_j|=|v_i+v_j| $ for all $ i,j $ then\\ $ (u_1,u_2,\ldots,u_k)=\pm(v_1,v_2,\ldots,v_k) $.
			\end{enumerate}
		\end{lemma} 
		\begin{proof} For the first part we have
			\begin{align*}
				|u_1-u_2|=|u_1+u_2|& \iff u_1-u_2= \pm (u_1+u_2) \\ 
				%& \phantom{\iff} \text{ or } u_1-u_2=-u_1-u_2\\
				& \iff 2u_2=0 \text{ or } 2u_1=0\\
				& \iff u_2=0 \text{ or } u_1=0 & \text{(since $ r $ is odd)}. 
			\end{align*}
			For the second part we note that the case of $ k=1 $ is patently true, so let $ k=2 $. If either $ u_1=0 $ or $ u_2=0 $ then the result holds, so assume otherwise. The hypothesis becomes  $ v_1=\pm u_1 $, $ v_2=\pm u_2 $, and   \[|u_1+u_2|=|v_1+v_2|=|\pm u_1 \pm u_2|,\] so  from part 1 we obtain $ (u_1,u_2)=\pm(v_1,v_2) $. Induction completes the proof.
		\end{proof}

		We now show that MWS Lee codes exist for all $ q $ and $ k $.
		\begin{theorem}\label{thm: Lee MWS bound determined}
			$ L_{\LL}(k,q) = \frac{q^k-1}{2} $, and $\displaystyle  M(\LL, k, q)\le \frac{ (\frac{q+1}{2})^{2k^2-3k+1}-1}{\frac{q+1}{2}-1} $. 
		\end{theorem}
		
		\begin{proof}
			For $ q=2 $ see \cite{MR4014640}. Assume $ q>2 $. 	Our proof is constructive. Let $ \alpha=\frac{q+1}{2} $ and define the generator matrix $ G_{k,q} $ by
			\[
			\left[ e_1\left|(e_2)^{\alpha}\left|\cdots \left|(e_k)^{\alpha^{k-1}}\left|\left(\sum_{i=1}^2 e_i\right)^{\alpha^{k}}\left| \left(\sum_{i=1}^3 e_i\right)^{\alpha^{k+1}}\left|\cdots\left| \left(\sum_{i=1}^k e_i\right)^{\alpha^{2k-2}}\right.\right.\right.\right.\right.\right.\right. \right]
			\]
			Recall, $ e_i $ denotes the unit column vector with $ 1 $ in entry $ i $ and exponents correspond to repetition numbers. The length of $ \C $ is $ n=\sum_{i=0}^t \alpha^i$ where $ t=2k-2 $. 
			From the form of $ G_{k,q} $ it follows that the Lee weight of any codeword $ uG_{k,q} $ may be expressed in base $ \alpha $  as 
			\[
			w_{\LL}(uG_{k,q}) = |u_1|+|u_2|\alpha +|u_3|\alpha^3 + \cdots +|u_1+u_{2} +\cdots +u_k|\alpha^{ 2k-2}.
			\]  
			If two codewords $ uG_{k,q} $ and $ vG_{k,q} $ have the same weight then $ |u_i|=|v_i| $, for all $ i $ and $ |u_1+u_2|=|v_1+v_2|$. From the Lemma \ref{lem: absolute values giving equality} we obtain $ (u_1,u_2)=\pm (v_1,v_2) $. Likewise, since $ |u_1+u_2+u_3|=|v_1+v_2+v_3|$  and $ |u_1+u_2|=|v_1+v_2|$, Lemma \ref{lem: absolute values giving equality} gives $ (u_1,u_2,u_3)=\pm (v_1,v_2,v_3) $. Continuing inductively we obtain $ \C $ is $ \LL $-MWS.\\
			
		\end{proof}

\begin{example}		
 If $ k=2 $ and $ q=5 $ then the generator matrix as in the proof of Theorem \ref{thm: Lee MWS bound determined} is  
		\[
		G_{2,5}= \left[ \begin{array}{c}
			1\\0
		\end{array}\begin{array}{c}
			0\\1
		\end{array}  
		\begin{array}{c}
			0\\1
		\end{array}  
		\begin{array}{c}
			0\\1
		\end{array}  
		\begin{array}{c}
			1\\1
		\end{array}\begin{array}{c}
			1\\1
		\end{array}\begin{array}{c}
			1\\1
		\end{array}\begin{array}{c}
			1\\1
		\end{array}\begin{array}{c}
			1\\1
		\end{array}\begin{array}{c}
			1\\1
		\end{array}\begin{array}{c}
			1\\1
		\end{array}\begin{array}{c}
			1\\1
		\end{array}\begin{array}{c}
			1\\1
		\end{array} \right].
		\] 
\end{example}

		\subsection{Lee FWS Codes} \label{subsec: LeeFWS}
		
		Recall from the definition that for $ q $ odd, an $ [n,k]_q $-code is $ \LL $-FWS if $ |w_{\LL}(\C)|=n\cdot \frac{q-1}{2} $. Table \ref{table: optimap bounds for Lee Metric} provides some examples of FWS codes for small values of $ q $ and $ k $.  Based simply on the properties of a component-wise metric we get an initial upper bound on the length of $ \LL $-FWS codes. Indeed, from part \ref{Part 2: lemma on bounds in general case} of Lemma \ref{lem: bounds on length in general case} and Proposition \ref{prop: Maximum Weight Spectrum and Full Weight Spectrum Codes}  we obtain $ n\le \frac{q^k-1}{q-1} $.
		By considering Theorem \ref{thm: ub2 length of General FWS} we can get a significantly better bound than this.

		\begin{lemma}\label{lem: LFWS necessary and sufficient}
			There exists an $ [n,k]_q $ $ \LL $-FWS code if and only if   \[n \le \sum\limits_{i=0}^{k-1} \left(\frac{q+1}{2}\right)^i \;\;\;\;\left(= \frac{\left(\frac{q+1}{2}\right)^k-1}{\left(\frac{q+1}{2}\right)-1}\right).\]
			In particular $ N(\LL,k,q)= \frac{\left(\frac{q+1}{2}\right)^k-1}{\left(\frac{q+1}{2}\right)-1}.$
		\end{lemma}
		\begin{proof}
			Follows from Theorem \ref{thm: ub2 length of General FWS} by observing that for the Lee metric $  m=\Delta  =\lceil (q-1)/2 \rceil $. 
		\end{proof}
		
		Comparing the bounds in the Lemmata \ref{lem: lb length of LMWS}  and \ref{lem: LFWS necessary and sufficient} we obtain the following non-existence result.
		
		\begin{corollary}
			There exists an $ [n,k]_q $ code that is both $ \mathcal L $-MWS and $ \mathcal{L} $-FWS if and only if $ q=2 $ .
		\end{corollary}
		\begin{proof}
			The Hamming and Lee metrics are the same for binary codes, so for the binary case see  \cite{MR4014640}.  The case $ q\ge 3 $ follows from Lemma \ref{lem: LFWS necessary and sufficient} and Corollary \ref{cor: FWS and MWS}.	
		\end{proof}

		\subsection{Manhattan MWS Codes} \label{subsec: MMWS}
		If $ u\in \F_q^n $ is a nonzero vector then $ | \langle u \rangle|= q $, so $ |w_{\MM}(\langle u \rangle)|\le q-1 $. One may show that strict inequality may occur only if $q$ is a divisor of $w_{\MM}(u)$. As such, an immediate bound on the $ \MM $-weight spectrum of a linear code is
		\begin{equation}\label{eqn: MMWS}
			L_{\MM}(k,q)\le  q^k-1.
		\end{equation}
		It is clear that for $ e_1 \in \F_q^n$, $ |w_M(\langle e_1 \rangle)|= q-1 $, so for $ n\ge1 $  we have   
		\begin{equation}\label{eqn: MMWS k=1}
			L_{M}(n,1,q)=L_{M}(1,q) = q-1.
		\end{equation}
		In \cite{BenMorineMScThesis2022} values of  $ L_{M}(n,k,q) $  were explicitly computed for $ q=3,5 $ and $ k=2 $, these values  appear in Table \ref{table: small values of L_M(n,k,q)}. Recall that for $ q=2 $ the Manhattan, Hamming, and Lee weights coincide, so we refer to section \ref{subsec: Existence of LMWS codes for small k,q}  for the binary case.   
		
		\begin{table}[h] 
			\caption{Optimal bounds for for the Manhattan Weight}
			\label{table: small values of L_M(n,k,q)}	
			\begin{center}
				\begin{tabular}{|c|c|c|c|}
					\hline
					Bound	& $ \MM $-MWS &  $ \MM $-FWS \\
					\hline
					$ L_{\MM}(3,2,3)=6 $&  No  & Yes  \\
					\hline
					$ L_{\MM}(4,2,3)=8 $&  Yes &   Yes  \\
					\hline
					$ L_{\MM}(3,2,5)=12 $&  No & 	 Yes  \\
					\hline
					$ L_{\MM}(4,2,5)=16 $&  No &	 Yes  \\
					\hline
					$ L_{\MM}(5,2,5)=20 $& No &  Yes  \\
					\hline
					$ L_{\MM}(6,2,5)=24 $&Yes &   Yes  \\
					\hline
				\end{tabular}
			\end{center}
		\end{table}
		
		Since the Manhattan weight is a coordinate-wise weight function, the maximum weight of an $ n $-vector is $ n(q-1) $. If $ \C $ is an $ [n,k]_q $ code then $ |w_\MM(\C)|\le \max w_\MM(\C)\le n(q-1) $, giving the following.    
		\begin{lemma}\label{lem: ub on MMWS length}
			If $ \C $ is an $ [n,k]_q $ code then $ n\ge \frac{|w_\MM(\C)|}{q-1} $. In particular, if $ \C $ is an $ [n,k]_q $ $ \MM $-MWS code then $ n\ge \frac{L_\MM(k,q)}{q-1} $ with equality if and only if $ w_\MM(\C)=\{1,2,\ldots,n(q-1)\} $.
		\end{lemma}
	 
		\begin{theorem}\label{thm: MMWS bound determined}
			For $ k\ge 1 $,	 \[ L_{\MM}(k,q) = q^k-1  ,\text{ and }   M(\MM, k, q)= \frac{q^k-1}{q-1}.\]
		\end{theorem}
		
		\begin{proof}
			For $ q=2 $ see Table \ref{table: small values of L_M(n,k,q)}. Assume $ q>2 $. 	Our proof is constructive. Define the generator matrix $ G(k,q) $ by	
			\[
			G_{k,q}=\left[ e_1\left|(e_2)^{q}\left|\cdots \left|(e_k)^{q^{k-1}}\right.\right.\right. \right]
			\]
			An argument entirely similar to that in the proof of Theorem \ref{thm: Lee MWS bound determined} shows that two codewords $ uG_{k,q} $ and $ vG_{k,q} $ have the same weight if and only if  $  u= v  $,  whence $ \C $ is $ \MM $-MWS. For the second part we observe that the length of the code generated by $ G $ is $ \frac{q^k-1}{q-1} $ and appeal to Lemma \ref{lem: ub on MMWS length}. \\
		\end{proof}

		\subsection{Manhattan FWS codes} \label{subsec: MFWS}
		It is readily verified that $ w_\MM(F^n_q)=n(q-1) $, so $ \MM $-FWS codes are those with $ w_\MM(\C)= n(q-1) $.  Table \ref{table: small values of L_M(n,k,q)} exhibits examples of such codes for small $ k,q $. The following shows that in contrast to MWS codes, we may completely determine the optimal lengths of FWS codes.  
		
		\begin{theorem}\label{thm: ub length of MFWS}
			There exists an $ [n,k]_q $ $ \MM $-FWS code if and only if $ k\le n\le \frac{q^k-1}{q-1} $. In particular, 	
			\[
			N(\MM, k, q) =\frac{q^k-1}{q-1}.
			\] 	
		\end{theorem}
		\begin{proof}
			That $ n\le \frac{q^k-1}{q-1}  $ follows directly from part \ref{Part 2: lemma on bounds in general case} of Lemma \ref{lem: bounds on length in general case} and Proposition \ref{prop: Maximum Weight Spectrum and Full Weight Spectrum Codes} by observing that $ m=q-1=\Delta $.  The codes with generator $ G_{k,q} $ as described in the proof of Theorem \ref{thm: MMWS bound determined} are $ [n,k]_q $-codes with $ w_\MM(\C) =\frac{q^k-1}{q-1}=n $ and are thus $ \MM $-FWS.\\
			Generator matrices obtained by deleting columns of $ G_{k,q} $ in the manner described in the proof of Theorem \ref{thm: General FWS necessary and sufficient} show the existence of $ \MM $-FWS codes for $ k\le n \le \frac{q^k-1}{q-1} $.  	
		\end{proof}

		\section{Conclusion, future work, and open problems}\label{sec: conclusion}
		Here we have investigated the combinatorial problem of maximizing the number of distinct weights that a linear code may obtain. The literature relating to this problem has thus far focused on FWS and MWS codes with respect to Hamming weight. Under the Hamming metric, FWS codes are necessarily ``short", and $ MWS $ codes are necessarily ``long".  The existence question and in particular the maximum length $ N(H,k,q) $ of an $ [n,k]_q $ FWS code has been determined. However,  the determination of $ M(H, k,q) $, the minimum length of an $ [n,k]_q $ MWS code remains an open problem.\\   
		Here, we have generalized to the setting of component-wise weight functions, with a particular focus on the Lee, and Manhattan type weights. For a certain class of component-wise weight functions (which subsumes the Hamming, Lee, and Manhattan weight functions) we determine the necessary and sufficient conditions for the existence of $ \mu $-FWS codes (Theorem \ref{thm: General FWS necessary and sufficient}).   We establish through construction  the existence of $ \mu $-MWS codes for all $ q,k $  (Theorem \ref{thm: Lee MWS bound determined}), and thus an upper bound on $ M(\LL, k,q) $. Computational examples show our bound on $ M(\LL, k,q) $ is not generally tight, thus the determination of $ M(\LL, k,q) $  remains an open problem.  On the other hand, for Manhattan weights we were able to determine both $ N(\MM, k, q) $ and $ M(\MM, k, q) $ and provided a construction showing that in fact  $ N(\MM, k, q) = M(\MM, k, q) $  (Theorem \ref{thm: ub length of MFWS}).\\
		The determination of $ M(H,k,q) $, $ M(\LL,k,q) $, and more generally $ M(\mu,k,q) $  (for an arbitrary component-wise weight function $ w_\mu $ on $ \F_q^n $) may be an attractive endeavour for the community.

\bibliographystyle{spmpsci}      % mathematics and physical sciences
%\bibliography{C:/Users/tlald/OneDrive/texed/my-texmf/bibtex/bib/main}   % name your BibTeX data base
%\bibliography{C:/Users/tlald/OneDrive/texed/my-texmf/bibtex/bib/main}   % name your BibTeX data base
\bibliography{../../../../my-texmf/bibtex/bib/main}   % name your BibTeX data base
%\bibliography{MWSFWSG}
%% Non-BibTeX users please use
%\begin{thebibliography}{}
%%
%% and use \bibitem to create references. Consult the Instructions
%% for authors for reference list style.
%\end{thebibliography}

\end{document}